\documentclass[conference,letterpaper]{IEEEtran}

\usepackage[utf8]{inputenc} 
\usepackage[T1]{fontenc}
\usepackage{url}
\usepackage{ifthen}
\usepackage{cite}
\usepackage[cmex10]{amsmath} 
                          
\usepackage{graphicx}
\usepackage{bm}

\usepackage{nidanfloat}
\usepackage{subcaption}

\usepackage{tikz}
\usetikzlibrary{matrix}

\usetikzlibrary{arrows}

\usepackage{amssymb}   
\usepackage{amsthm}

\newcommand{\argmax}{\operatornamewithlimits{argmax}} 
\newcommand{\Inf}{\operatorname{Inf}}
\newcommand{\Var}{\operatorname{Var}}  
\newcommand{\Maj}{\operatorname{Maj}}

\newtheorem{theorem}{Theorem}
\newtheorem{restate}{Theorem}
\newtheorem{proposition}{Proposition}
\newtheorem{lemma}{Lemma}
\newtheorem{corollary}{Corollary}

\theoremstyle{definition}
\newtheorem{example}{Example}

\theoremstyle{definition}
\newtheorem*{notation}{Notation}

\theoremstyle{definition}
\newtheorem{definition}{Definition}

\theoremstyle{definition}
\newtheorem*{remark}{Remark}

\newtheorem{hypothesis}{Hypothesis}

\begin{document}
\title{\vspace{-4pt}The Computational Wiretap Channel} 


                    

\author{%
   \IEEEauthorblockN{Rafael G.L. D'Oliveira\IEEEauthorrefmark{1},
                     Salim El Rouayheb\IEEEauthorrefmark{3},
                     and Muriel Médard\IEEEauthorrefmark{4}}
   
   \IEEEauthorblockA{\IEEEauthorrefmark{1}%
                     Department of Electrical and Computer Engineering,
                     Rutgers University,
                     rafael.doliveira@rutgers.edu}
                     
   \IEEEauthorblockA{\IEEEauthorrefmark{3}%
                     Department of Electrical and Computer Engineering,
                     Rutgers University,
                     salim.elrouayheb@rutgers.edu}
   \IEEEauthorblockA{\IEEEauthorrefmark{4}%
                     Engineering, Electrical Engineering and Computer Science, Massachusetts Institute of Technology,
                     medard@mit.edu }
 }

\maketitle

\begin{abstract}
  We present the computational wiretap channel: Alice has some data $x$ and wants to share some computation $h(x)$ with Bob. To do this, she sends $f(x)$, where $f$ is some sufficient statistic for $h$. An eavesdropper, Eve, is interested in computing another function $g(x)$. We show that, under some conditions on $f$ and $g$, this channel can be approximated, from Eve's point of view, by the classic Wyner wiretap channel.
\end{abstract}

\section{Introduction}

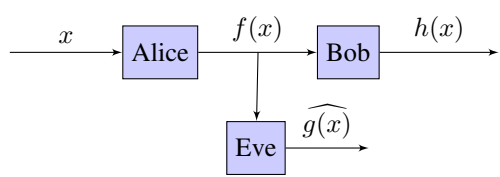
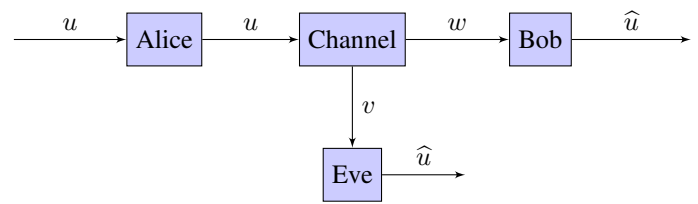
\begin{figure*}[!b]
    \begin{subfigure}[b]{0.5\linewidth}
    \centering

\tikzstyle{int}=[draw, fill=blue!20, minimum size=2em]
\tikzstyle{init} = [pin edge={to-,thin,black}]

\begin{tikzpicture}[node distance=2.5cm,auto,>=latex']
    \node [int] (a) {Alice};
    \node (b) [left of=a,node distance=2cm, coordinate] {a};
    \node [int] (c) [right of=a] {Bob};
    \node  (e) [right of=a, node distance=1.3cm] {};
    \node [coordinate] (end) [right of=c, node distance=2cm]{};
    \node [int] (d) [below right of=a,node distance=1.8cm] {Eve};
    \node [coordinate] (end2) [right of=d, node distance=1.5cm]{};
    \path[->] (b) edge node {$x$} (a);
    \path[->] (a) edge node {$f(x)$} (c);
    \path[->] (e.center) edge node {} (d);
    \draw[->] (c) edge node {$h(x)$} (end) ;
    \draw[->] (d) edge node {$\widehat{g(x)}$} (end2) ;
\end{tikzpicture}
\vspace{2ex}
\caption{The Computational Wiretap Channel: Alice sends a computation $f(x)$ of her data, where $f$ is a sufficient statistic for $h$. Thus, Bob is able to retrieve $h(x)$ perfectly and Eve has an estimate $\widehat{g(x)}$ of $g(x)$.} \label{fig:comp} 

    \vspace{3.5ex}
  \end{subfigure} \hspace{5pt}
  \begin{subfigure}[b]{0.5\linewidth}
    \centering

\tikzstyle{int}=[draw, fill=blue!20, minimum size=2em]
\tikzstyle{init} = [pin edge={to-,thin,black}]

\begin{tikzpicture}[node distance=2.5cm,auto,>=latex']
    \node [int] (a) {Alice};
    \node (b) [left of=a,node distance=2cm, coordinate] {a};
    \node [int] (c) [right of=a] {Channel};
    \node [int] (d) [right of=c] {Bob};
    \node [coordinate] (end1) [right of=d, node distance=2cm]{};
    
    \node [int] (e) [below of=c,node distance=1.8cm] {Eve};
    \node [coordinate] (end2) [right of=e, node distance=1.5cm]{};
    
    \path[->] (b) edge node {$u$} (a);
    \path[->] (a) edge node {$u$} (c);
    \path[->] (c) edge node {$w$} (d);
    \draw[->] (d) edge node {$\widehat{u}$} (end1);
    
    \path[->] (c) edge node {$v$} (e);
    \draw[->] (e) edge node {$\widehat{u}$} (end2);
    
\end{tikzpicture}
\caption{The Classic Wiretap Channel: Alice sends a message $u \in U$ through a channel which outputs $w \in W$ to Bob and $v \in V$ to Eve with probability $\Pr(\bm{w}=w,\bm{v}=v | \bm{u}=u)$. Bob and Eve estimate $u$.} \label{fig:genwire}
    \vspace{4ex}
  \end{subfigure} 
  
  \caption{Computational vs.\ classic wiretap channel. Our main result is that, under certain conditions on $f$ and $g$, the computational wiretap channel can be approximated, from Eve's point of view, to a classic wiretap channel.}
  \label{fig7} 
\end{figure*}

We present the computational wiretap channel. Alice has some data $x$ and wants to share some computation $h(x)$ with Bob. To do this, she sends $f(x)$, where $f$ is some sufficient statistic for $h$. An eavesdropper, Eve, is interested in computing another function $g(x)$. A diagram for this channel is shown in Figure \ref{fig:comp}.

The computational wiretap channel is a natural model for various settings. For example, Alice could be a user in a social network sharing articles, pictures or videos she likes  with her friend, Bob, and Eve could be the service provider trying to classify some of Alice's personal attributes like her sexual orientation, ethnicity, political views, etc (See \cite{kosi13}). 

Our main result is that, under certain conditions on the functions $f$ and $g$, the computational wiretap channel can be approximated by the classic wiretap channel \cite{wyne75}, shown in Figure \ref{fig:genwire}. Our result has two versions, one for real-valued Boolean functions and one for Boolean functions. We state them here informally.

\begin{theorem}{(Informal)}
Let $f,g: \{-1,1\}^n \rightarrow \mathbb{R}$ be real-valued Boolean functions. Suppose $f$ and $g$ are low influence functions, i.e., their values do not rely too much on any coordinate. Then, from Eve's point of view, the computational wiretap, shown in Figure \ref{fig:comp}, can be approximated by an additive wiretap channel, shown in Figure \ref{fig:addit}.
\end{theorem}

\begin{theorem}{(Informal)}
Let $f,g: \{-1,1\}^n \rightarrow \{-1,1\}$ be Boolean functions. Suppose $f$ and $g$ are low influence functions, i.e., their values do not rely too much on any coordinate. Then, from Eve's point of view, the computational wiretap, shown in Figure \ref{fig:comp}, can be approximated by a multiplicative wiretap channel, shown in Figure \ref{fig:mult}.
\end{theorem}

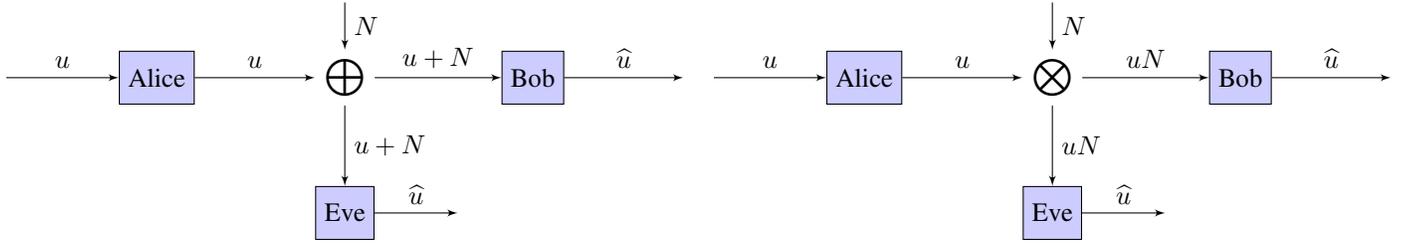
\begin{figure*}[!t]
  \begin{subfigure}[b]{0.5\linewidth}
    \centering

\tikzstyle{int}=[draw, fill=blue!20, minimum size=2em]
\tikzstyle{init} = [pin edge={to-,thin,black}]

\begin{tikzpicture}[node distance=2.5cm,auto,>=latex']
    \node [int] (a) {Alice};
    \node (b) [left of=a,node distance=2cm, coordinate] {a};
    \node (c) [right of=a] {\Large{$\bigoplus$}};
    \node [int] (d) [right of=c] {Bob};
    \node [coordinate] (end1) [right of=d, node distance=2cm]{};
    
    \node [int] (e) [below of=c,node distance=1.8cm] {Eve};
    \node [coordinate] (end2) [right of=e, node distance=1.5cm]{};
    
    \node (f) [above of=c,node distance=1cm, coordinate] {c};
    
    \path[->] (b) edge node {$u$} (a);
    \path[->] (a) edge node {$u$} (c);
    \path[->] (c) edge node {$u + N$} (d);
    \draw[->] (d) edge node {$\widehat{u}$} (end1);
    
    \path[->] (c) edge node {$u + N$} (e);
    \draw[->] (e) edge node {$\widehat{u}$} (end2);
    
    \path[->] (f) edge node {$N$} (c);

\end{tikzpicture}
\caption{The Additive Wiretap Channel: Alice sends a message $u \in \mathbb{R}$ through a channel which outputs $u+N \in \mathbb{R}$. Both Bob and Eve output an estimate $\widehat{u}$ of $u$.} \label{fig:addit}
  \end{subfigure} \hspace{5pt}
  \begin{subfigure}[b]{0.5\linewidth}
    \centering

\tikzstyle{int}=[draw, fill=blue!20, minimum size=2em]
\tikzstyle{init} = [pin edge={to-,thin,black}]

\tikzstyle{int}=[draw, fill=blue!20, minimum size=2em]
\tikzstyle{init} = [pin edge={to-,thin,black}]

\begin{tikzpicture}[node distance=2.5cm,auto,>=latex']
    \node [int] (a) {Alice};
    \node (b) [left of=a,node distance=2cm, coordinate] {a};
    \node (c) [right of=a] {\Large{$\bigotimes$}};
    \node [int] (d) [right of=c] {Bob};
    \node [coordinate] (end1) [right of=d, node distance=2cm]{};
    
    \node [int] (e) [below of=c,node distance=1.8cm] {Eve};
    \node [coordinate] (end2) [right of=e, node distance=1.5cm]{};
    
    \node (f) [above of=c,node distance=1cm, coordinate] {c};
    
    \path[->] (b) edge node {$u$} (a);
    \path[->] (a) edge node {$u$} (c);
    \path[->] (c) edge node {$u  N$} (d);
    \draw[->] (d) edge node {$\widehat{u}$} (end1);
    
    \path[->] (c) edge node {$u  N$} (e);
    \draw[->] (e) edge node {$\widehat{u}$} (end2);
    
    \path[->] (f) edge node {$N$} (c);

\end{tikzpicture}
\caption{The Multiplicative Wiretap Channel: Alice sends a message ${u \in \{-1,1\}}$ through a channel which outputs $uN \in \{-1,1\}$. Both Bob and Eve output an estimate $\widehat{u}$ of $u$.} \label{fig:mult}
  \end{subfigure} 
  \caption{Theorem \ref{teo:additive invariance} states that if $f$ and $g$ are real-valued Boolean functions satisfying some conditions then, from Eve's point of view, the computational wiretap can be approximated by an additive wiretap channel. Theorem \ref{teo: mult invariance} states that in the case of Boolean functions the computational wiretap can be approximated by a multiplicative wiretap channel.  }
  \label{fig8} 
\end{figure*}

The proofs of Theorems \ref{teo:additive invariance} and \ref{teo: mult invariance} rely heavily on a generalization, known as the Basic Invariance Principle \cite{moss10}, of the Berry-Esseen theorem. To use this result we will need some notation and tools from the field of analysis of Boolean functions, the topic of Section \ref{sec:analysisboolfun}.

In the remaining sections we consider different classes for the functions $f$ and $g$. For each class of functions we show a formal equivalence between the computational wiretap channel and the classic wiretap channel in  Theorems \ref{teo:genequi}, \ref{teo:additive equi}, and \ref{teo:mult equivalence}. These formal equivalences are used to prove our two main results, Theorem \ref{teo:additive invariance} in Section \ref{sec:additive}, and Theorem \ref{teo: mult invariance} in Section \ref{sec:multiplicative}.

\section{Analysis of Boolean Functions} \label{sec:analysisboolfun}

In this section we give the necessary tools for proving Theorems \ref{teo:additive invariance} and \ref{teo: mult invariance}. These theorems rely heavily on what is known as the Basic Invariance Principle \cite{moss10}, presented in Theorem \ref{teo:invprince}, a generalization of the Berry-Esseen Theorem.

All results in this section, apart from Lemmas~\ref{lem:var difference}, \ref{lem:influence diff}, and \ref{lem: influence mult}, are taken from \cite{odon14} and are included here for the  convenience of the reader.

Analysis of Boolean functions is the study of real-valued Boolean functions $f:\{-1, 1\}^n \rightarrow \mathbb{R}$ using analytical techniques. We begin by looking at the Fourier expansion.

Every real-valued Boolean function, $f:\{-1,1\}^n \rightarrow \mathbb{R}$ can be represented as a real multilinear polynomial, known as the Fourier expansion of $f$.

\begin{theorem} \label{teo: fourier representation}
Every function $f:\{-1,1\}^n \rightarrow \mathbb{R}$ can be uniquely expressed as a multilinear polynomial,
\[ f(x) = \sum_{S \subseteq [n]} \widehat{f}(S) x^S \]
where $\widehat{f}:2^{[n]} \rightarrow \mathbb{R}$ and $x^S = x_1^{i_1} \ldots x_n^{i_n}$ with $i_k=1$ if $k \in S$ and $i_k=0$ otherwise.		
\end{theorem}
\begin{proof}
Theorem 1.1 in \cite{odon14}.
\end{proof}

\begin{example} \label{exe:majority fourier expansion}
The majority function on $3$ bits, denoted by ${\Maj_3 : \{-1,1\}^3 \rightarrow \{-1,1\}}$, outputs the most frequent $\pm 1$ bit in the input. It is easy to check that its Fourier expansion is
\[ \Maj_3 (x_1,x_2,x_3) = \frac{1}{2} \left( x_1 + x_2 + x_3 - x_1 x_2 x_3 \right) .\]
\end{example}

Our main results will rely on the notion of the influence of a coordinate. This notion was originally introduced in \cite{penr46} in the context of social choice theory and has found many other uses in combinatorics and computer science.

\begin{notation}
We will always write random variables in boldface. Probabilities and expectations will always be with respect to a uniformly random $\bm{x} \sim \{-1,1\}^n$ unless specified otherwise.
\end{notation}

\begin{definition} \label{def: influence}
The \emph{influence} of coordinate $i$ in the function $f:\{-1,1\}^n \rightarrow \{-1,1\}$ is defined as \[ {Inf}_i [f] = \Pr [f(\bm{x}) \neq f(\bm{x}^{\oplus i})] \]
where $x^{\oplus i} = (x_1,\ldots,x_{i-1}, -x_i,x_{i+1},\ldots, x_n)$.
\end{definition}

Informally, the influence of a coordinate measures how much it influences the value of the function. 

The influence can be expressed in terms of the function's Fourier expansion, allowing Definition \ref{def: influence} to be extended to real-valued Boolean functions.

\begin{theorem} \label{teo: influence}
Let $f:\{-1,1\}^n \rightarrow \mathbb{R}$ and $i \in [n]$. Then, 
\[ {Inf}_i [f] = \sum_{S \ni i} \widehat{f}(S)^2 .\]
\end{theorem}
\begin{proof}
Theorem 2.20 in \cite{odon14}.
\end{proof}

\begin{example} \label{exe:majority influence}
Consider the majority function on $3$ bits in Example \ref{exe:majority fourier expansion}. Then, for every $t \in [3]$, $\Inf_t [\Maj_3] = 1/2$.
\end{example}

Another key property which can be expressed in terms of the function's Fourier expansion is the variance.

\begin{proposition}
The variance of $f: \{-1,1\}^n \rightarrow \mathbb{R}$ is
\[ \Var [f] = \sum_{S \neq \emptyset} \widehat{f}(S)^2 \]
\end{proposition}

\begin{proof}
Proposition 1.13 in \cite{odon14}.
\end{proof}

\begin{example} \label{exe:majority variance}
Consider the majority function on $3$ bits in Example \ref{exe:majority fourier expansion}. Then, $Var [\Maj_3] = 1$.
\end{example}

The basic invariance principle, Theorem \ref{teo:invprince}, gives conditions under which the random variable $\bm{x} = (\bm{x_1},\ldots, \bm{x_n} )$ can be substituted by $\bm{g} = (\bm{g_1},\ldots, \bm{g_n} )$, where each $\bm{g_i}$ is a standard Gaussian, i.e. a Gaussian with mean $0$ and variance $1$. We need the following hypothesis on the probability distributions.

\begin{hypothesis} \label{hyp:var}
The random variable $\bm{x}_i$ satisfies $E[\bm{x}_i]=0$, $E[\bm{x}_i^2]=1$, $E[\bm{x}_i^3]=0$, and $E[\bm{x}_i^4]\leq 9$.
\end{hypothesis}

The main examples to keep in mind are the uniform $\pm 1$ random bit and the standard Gaussian.

We now present the basic invariance principle.

\begin{theorem} \label{teo:invprince}
Let $F$ be a formal $n$-variate multilinear polynomial of degree at most $k \in \mathbb{N}$, 
\[ F(x) = \sum_{S \subseteq [n], |S| \leq k} \hat{F}(S) x^S .\]
Let $\bm{x}=(\bm{x}_1,\ldots,\bm{x}_n)$ and $\bm{y}=(\bm{y}_1,\ldots,\bm{y}_n)$ be sequences of independent random variables, each satisfying Hypothesis \ref{hyp:var}. Assume $\psi: \mathbb{R} \rightarrow \mathbb{R}$ is $\mathcal{C}^4$ with $||\psi''''||_\infty \leq C$.\footnote{Being $\mathcal{C}^4$ means that the derivatives $\psi', \ldots, \psi^{''''}$ exist and are continuous.} Then \[ | E[\psi(F(\bm{x}))]-E[\psi(F(\bm{y}))]| \leq \frac{C}{12} 9^k \sum_{t=1}^n {Inf}_t[F]^2 .\]
\end{theorem}
\begin{proof}
See page 357 in \cite{odon14}.
\end{proof}
Some things to note:
\begin{itemize}
    \item Since the standard Gaussian, $\bm{g}_i$, satisfies Hypothesis \ref{hyp:var}, $\bm{y}$ can be taken equal to $\bm{g}=(\bm{g}_1, \ldots, \bm{g}_n)$.
    \item The function $\psi$ is known as a test function. In applications, the test functions of interest might not be differentiable or bounded by their fourth derivative. However, these can often be approximated by smooth functions which do satisfy the necessary conditions. In Corollary 11.68 of \cite{odon14}, for example, the smoothness of $\psi$ is substituted by a Lipschitz condition.
    
    \item The goodness of the approximation depends on the degree, $k$, of the Fourier expansion of $f$. This requirement might be loosened by truncating the polynomial to a certain degree. This is done in Corollary 11.69 of \cite{odon14}.
\end{itemize}

In Theorems \ref{teo:additive invariance} and \ref{teo: mult invariance} we use the following corollary.

\begin{corollary} \label{cor:bip}
In the setting of Theorem \ref{teo:invprince}, if we furthermore have $Var[F]\leq 1$ and ${Inf}_t \leq \epsilon$ for all $t \in [n]$, then
\[ | E[\psi(F(\bm{x}))] - E[F(\psi(\bm{y}))] | \leq \frac{C}{12}  k 9^k \epsilon .\]
\end{corollary}

\begin{proof}
Corollary 11.67 in \cite{odon14}.
\end{proof}

\begin{example}
Consider the majority function on $3$ bits discussed in Examples \ref{exe:majority fourier expansion}, \ref{exe:majority influence}, and \ref{exe:majority variance}. In the terms of Corollary \ref{cor:bip} we have $k = 3$ and $\epsilon = 1/2$, so that
\[ | E[\psi( \Maj_3 (\bm{x}))] - E[\Maj_3 (\psi(\bm{g}))] | \leq 92 C .\]
\end{example}

We need three lemmas for Theorems \ref{teo:additive invariance} and \ref{teo: mult invariance}.

The first lemma bounds the variance of the difference of two real-valued Boolean functions.

\begin{lemma} \label{lem:var difference}
Let $f,g:\{-1,1\}^n \rightarrow \mathbb{R}$ be two real-valued Boolean functions. If both $Var[f]$ and $Var[g]$ are smaller than $1/4$, then ${Var[f-g] \leq 1}$ .
\end{lemma}

\begin{proof}
It follows from the Cauchy–Schwarz inequality that
\[ \left| Cov[f,g] \right|^2 \leq Var[f] Var[g] \leq \dfrac{1}{16} .\]
Thus,
\[ Var[f-g] = Var[f]-Var[g]-2 Cov[f,g] \leq 1 .\]
\end{proof}

The second lemma shows that if a variable has low influence on two functions, then it has low influence on their difference.

\begin{lemma} \label{lem:influence diff}
Let $f,g:\{-1,1\}^n \rightarrow \mathbb{R}$ be two real-valued Boolean functions. If both $\Inf_t [f]$ and $\Inf_t [g]$ are smaller than $\epsilon$, then ${\Inf_t [f-g] \leq 4 \epsilon}$.
\end{lemma}

\begin{proof}
\begin{align*}
\Inf_t [f-g] &= \sum_{S \ni t} \widehat{(f-g)}(S)^2
= \sum_{S \ni t} \left( \widehat{f}(S) - \widehat{g}(S) \right)^2 \\
&= \sum_{S \ni t} \hat{f}(S)^2 + \sum_{S \ni t} \hat{g}(S)^2 - 2 \sum_{S \ni t} \hat{f}(S) \hat{g}(S) \\
&\leq \epsilon+\epsilon+2\epsilon
\end{align*}
where $\sum_{S \ni t} \hat{f}(S) \hat{g}(S) \leq \epsilon$ follows from the Cauchy–Schwarz inequality.
\end{proof}

The third lemma shows that if a variable has low influence on two functions, then it has low influence on their multiplication, but depending on the number of terms in the polynomial.

\begin{lemma} \label{lem: influence mult}
Let $f,g:\{-1,1\}^n \rightarrow \{-1,1\}$ be two Boolean functions where the polynomials $f$ and $g$ have, respectively, $l_1$ and $l_2$ terms. If both $\Inf_t [f]$ and $\Inf_t [g]$ are smaller than $\epsilon$, then ${\Inf_t [fg] \leq 4 \epsilon l}$, where $l = l_1 l_2$.
\end{lemma}

\begin{proof}
By definition,
\[ \Inf_t [f.g] = \sum_{S \ni t} \widehat{fg}(S)^2 .\]
We can write $f = x_t q^f + r^f$ and $g= x_t q^g + r^g$ where $q^f,r^f,q^g$ and $r^g$ are polynomials which do not depend on $x_t$. Then, \[fg = x_t q^f r^g + x_t q^g r^f + q^f q^g + r^f r^g .\]

Since $q^f q^g + r^f r^g$ does not depend on $x_t$, \[Inf_t [fg] = Inf_t[x_t q^f r^g + x_t q^g r^f] .\]

We will first calculate $Inf_t[x_t q^f r^g]$. Since $f$ and $g$ have, respectively, $l_1$ and $l_2$ terms, there exists $S^f_i, S^g_j \subseteq [n]-\{t\}$, for every $i \in [l_1]$ and $j \in [l_2]$, such that
\[ q^f = \sum_{i=1}^{l_1} q^f_i x^{S_i} \hspace{5pt} \text{and} \hspace{5pt} r^g = \sum_{j=1}^{l_2} r^g_j x^{S_j} .\]
Thus, 
\[ x_t q^f r^g = \sum_{i=1}^{l_1} \sum_{j=1}^{l_2} q^f_i r^g_j x^{S_i + S_j} x_t .\]
Therefore,
\begin{align*}
    \Inf_t [x_t q^f r^g] &= \sum_{S \ni t} \widehat{x_t q^f r^g}(S)^2 \leq \left( \sum_{i=1}^{l_1} \sum_{i=1}^{l_2} q^f_i r^g_j \right)^2 \\
    &\leq l \sum_{i=1}^{l_1} \sum_{j=1}^{l_2} \left( q^f_i r^g_j \right)^2 = l \sum_{i=1}^{l_1} \left( q^f_i  \right)^2 \sum_{j=1}^{l_2} \left( r^g_j \right)^2 \\
    &\leq l \epsilon
\end{align*} 
where the second inequality follows from the Cauchy–Schwarz inequality and the last one from the fact that \[ \Inf_t [f] = \sum_{i=1}^{l_1} \left( q^f_i  \right)^2 \hspace{5pt} \text{and} \hspace{5pt} \sum_{j=1}^{l_2} \left( r^g_j \right)^2 \leq \Var [g] \leq 1 .\]

The same arguments can be used to show that \[ \Inf_t [x_t q^g r^f] \leq l \epsilon .\]

Thus, by Lemma \ref{lem:influence diff}, $\Inf_t[x_t q^f r^g + x_t q^g r^f] \leq 4 \epsilon l$.

\end{proof}

\section{General Functions} \label{sec:general}

In this section we show a formal equivalence, from Eve's point of view, between the computational wiretap channel, as in Figure \ref{fig:comp}, and the classic wiretap channel, as in Figure \ref{fig:genwire}.

Let $X$ be a set and $f: X \rightarrow Y$ and $g: X \rightarrow Z$ be two functions. We assume some probability distribution on $X$ which induces distributions on $Y$ and $Z$. We denote the corresponding random variables by $\bm{x}$, $\bm{y}$, and $\bm{z}$.

In the computational wiretap channel, shown in Figure~\ref{fig:comp}, Alice has some data $x\in X$ and wants to share some computation $h(x)$ with Bob. To do this, she sends $f(x)$, where $f$ is some sufficient statistic for $h$. An eavesdropper, Eve, is interested in computing another function $g(x)$.

In general, $f$ is not sufficient for computing $g$, and therefore Eve will have an estimate 
\[ \widehat{g(x)} = \argmax_{z\in Z} \Pr(\bm{z} = z | \bm{y} = f(x)). \]

In the classic wiretap channel, shown in Figure \ref{fig:genwire}, Alice sends a message $u \in U$ to Bob. A channel outputs noisy versions, $w \in W$ and $v \in V$, of $u$ to Bob and Eve with probability $\Pr(\bm{w}=w,\bm{v}=v | \bm{u}=u)$.

Eve`s estimate of $u$ is then
\[ \widehat{u} = \argmax_{u\in U} \Pr(\text{channel input} \hspace{3pt} u|\hspace{3pt} \text{channel output} \hspace{3pt} v) .\]

In Theorem \ref{teo:genequi} we show that, from Eve's point of view, the computational wiretap channel is indistinguishable from a classic wiretap channel. We need the following definition to make the statement precise.

\begin{definition}
Two channels are equivalent, from Eve's point of view, if for both channels, the distribution on Eve's input and output are the same.
\end{definition}

\begin{theorem} \label{teo:genequi}
Every computational wiretap channel is equivalent, from Eve's point of view, to a classic wiretap channel.
\end{theorem}

\begin{proof}
Consider the computational wiretap channel in Figure \ref{fig:comp}. Let $\bm{u}=g(\bm{x})$ and $\bm{v}=f(\bm{x})$ with probability distributions induced by $\bm{x}$. Let these define the variables in Figure \ref{fig:genwire}, the random variable $\bm{w}$ being immaterial (for definiteness set $\bm{w}=\bm{v}$). Then, in both channels, Eve's input is $f(\bm{x})=\bm{v}$ and her output is $\widehat{g(x)}=\widehat{u}$.
\end{proof}

Theorem \ref{teo:genequi} shows that although there is no noise in the computational wiretap channel, the function $f$ can be interpreted as a noisy version of $g$. We now show conditions on $f$ and $g$ under which $g$ can be retreived exactly from $f$.

To every input $f(x)$ received by Eve, there corresponds an estimate $\widehat{g(x)}$. By defining $\widehat{g}:Y \rightarrow Z$ as $\widehat{g} (f(x)) = \widehat{g(x)}$ we have the diagram in Figure \ref{fig:diagram}.

\begin{figure}[h] 
\center
\begin{tikzpicture}
  \matrix (m) [matrix of math nodes,row sep=3em,column sep=4em,minimum width=2em]
  {
     X & Y \\
      & Z \\};
  \path[-stealth]
    (m-1-1) edge node [above] {$f$} (m-1-2)
    (m-1-2) edge node [right] {$\widehat{g}$} (m-2-2)
    (m-1-1) edge node [below] {$g$} (m-2-2);
\end{tikzpicture}
\caption{Diagram for the functions $f$, $g$, and $\widehat{g}$.}\label{fig:diagram}
\end{figure}
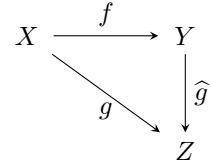

The diagram commutes if and only if the estimate is always correct, i.e. if for every $x \in X$ it follows that $\widehat{g(x)} = g(x)$. This occurs, for example, if $f$ is injective. In this case, $f$ has a left inverse $f^{-1}$ and by taking $\widehat{g} = g \circ f^{-1}$ it follows that $\widehat{g(x)} = (\widehat{g} \circ f) (x) = (g \circ f^{-1} \circ f) (x) = g(x)$ for every $x \in X$. 

The following result completely characterizes the commutativity of the diagram.

\begin{proposition} \label{pro:diacomm}
The diagram in Figure \ref{fig:diagram} commutes if and only if $f(x) = f(x')$ implies that $g(x)=g(x')$.
\end{proposition}

\begin{proof}
Suppose that the diagram commutes. Let $x,x' \in X$ be such that $f(x) = f(x')$. Then, \[ g(x) = (\widehat{g} \circ f)(x) = (\widehat{g} \circ f)(x') = g(x')\]

For the converse, suppose that $f(x) = f(x')$ implies that $g(x)=g(x')$ and let $y \in Y$. The fiber of $y$ by $f$ is the set $(f^{-1}(y) = \{x \in X : f(x) = y\}$. It follows from our hypothesis that the image  $g\circ f^{-1}(y)$ has a single element $z \in Z$. Thus  $\widehat{g}(y) = z$.
\end{proof}

Proposition \ref{pro:diacomm} can be restated as follows.

\begin{corollary}
A computational wiretap channel is equivalent to a noiseless wiretap channel if and only if, for every $x \in X$, $f(x) = f(x')$ implies that $g(x)=g(x')$.
\end{corollary}

\section{Real-Valued Boolean Functions} \label{sec:additive}

In this section, we show that when the functions $f,g: \{-1,1\}^n \rightarrow \mathbb{R}$ are real-valued Boolean functions with low influence, the computational wiretap channel, in Figure \ref{fig:comp}, can be approximated, from Eve's point of view, to an additive wiretap channel, as in Figure \ref{fig:addit}. 

We begin by showing a formal equivalence between these two channels, analogous to Theorem \ref{teo:genequi}.

\begin{theorem} \label{teo:additive equi}
Every computational wiretap channel is equivalent, from Eve's point of view, to an additive wiretap channel.
\end{theorem}

\begin{proof}
Consider the computational wiretap channel in Figure \ref{fig:comp}. Let $\bm{u}=g(\bm{x})$ and $N(\bm{x}) = f(\bm{x}) - g(\bm{x})$ define the variables in Figure \ref{fig:addit}. Then, in both channels, Eve's input is $f(\bm{x})=u(\bm{x})+N(\bm{x})$ and her output is $\widehat{g(x)}=\widehat{u}$.
\end{proof}

The noise, $N(\bm{x})$, of the additive wiretap channel, in Theorem \ref{teo:additive equi}, depends on the data $\bm{x}$ and the functions $f$ and $g$. In the next theorem, we show that if $\bm{x}$ is a well behaved random variable and $f$ and $g$ are low influence functions, i.e. their values do not depend too much on any coordinate, then $N(\bm{x})$ can be approximated by some noise, $N(\bm{g})$, which only depends on the functions $f$ and $g$ and is independent of $\bm{x}$.

\begin{restate} \label{teo:additive invariance}
Let $\bm{x} \sim \{-1 , 1\}^n$ satisfy Hypothesis \ref{hyp:var} and $\bm{g} = (\bm{g}_1, \ldots, \bm{g}_n)$ be such that each $\bm{g}_i$ is a standard Gaussian. Let $f,g:\{-1,1\}^n \rightarrow \mathbb{R}$ be of degree $k_1$ and $k_2$ with both $\Inf_t [f]$ and $\Inf_t [g]$ smaller than $\epsilon$, for every $t \in [n]$, and both  $\Var[f]$ and $\Var[g]$ smaller than $1/4$. Assume $\psi: \mathbb{R} \rightarrow \mathbb{R}$ is $\mathcal{C}^4$ with $||\psi''''||_\infty \leq C$. Then the noise $N=f-g$ satisfies
\[ | E[\psi(N(\bm{x}))] - E[\psi(N(\bm{g}))] | \leq \frac{C}{3}  k 9^k  \epsilon \]
where $k=k_1 k_2$.
\end{restate}

\begin{proof}
Since $\Var[f]$ and $\Var[g]$ are smaller than $1/4$, Lemma \ref{lem:var difference} implies in ${\Var[N] \leq 1}$. Since $\Inf_t [f]$ and $\Inf_t [g]$ are smaller than $\epsilon$, Lemma \ref{lem:influence diff} implies in $\Inf_t [N]\leq 4 \epsilon$. Our result then follows from Corollary \ref{cor:bip}. 
\end{proof}

\begin{example}
Consider the computational wiretap in Figure \ref{fig:comp} where $f,g: \{1,-1 \}^n \rightarrow \mathbb{R}$ are such that
\[ f(x) = \frac{1}{n} \sum_{i=1}^{n-1} x_i x_{i+1} \hspace{5pt} \text{and} \hspace{5pt} g(x) = \frac{1}{n} \sum_{i=1}^{n} x_i .\]

By Theorem \ref{teo:additive equi}, this computational wiretap channel is equivalent, from Eve's point of view, to the additive wiretap channel in Figure \ref{fig:addit} with noise $N(\bm{x}) = f(\bm{x}) - g(\bm{x})$. 

This noise depends not only on $f$ and $g$ but also on $\bm{x}$. Using Theorem $\ref{teo:additive invariance}$ we can approximate this noise by $N(\bm{g})$ which is independent of $\bm{x}$. 

One can check that both $\Inf_t [f]$ and $\Inf_t [g]$ are smaller than $2/n^2$, both $Var[f]$ and $Var[g]$ are smaller than $1/n$, and $k=2$. Thus,
\[ | E[\psi(N(\bm{x}))] - E[\psi(N(\bm{g}))] | \leq \frac{108 C}{n^2} .\]
\end{example}

\section{Boolean Functions} \label{sec:multiplicative}

In this section, we show that when the functions $f,g: \{-1,1\}^n \rightarrow \{-1,1\}$ are Boolean functions with low influence, the computational wiretap channel, in Figure \ref{fig:comp}, can be approximated, from Eve's point of view, to a multiplicative wiretap channel, as in Figure \ref{fig:mult}. 

We begin by showing a formal equivalence between these two channels, analogous to Theorems \ref{teo:genequi} and \ref{teo:additive equi}.

\begin{theorem} \label{teo:mult equivalence}
Every computational wiretap channel is equivalent, from Eve's point of view, to a multiplicative wiretap channel.
\end{theorem}

\begin{proof}
Consider the computational wiretap channel in Figure \ref{fig:comp}. Let $\bm{u}=g(\bm{x})$ and $N(\bm{x}) = f(\bm{x})g(\bm{x})$ define the variables in Figure \ref{fig:mult}. Then, in both channels, Eve's input is ${f(\bm{x})=u(\bm{x})N(\bm{x})}$ and her output is $\widehat{g(x)}=\widehat{u}$.
\end{proof}

\begin{remark}
The multiplicative wiretap channel in Figure \ref{fig:addit} is equivalent to a binary asymmetric channel with probability $\Pr(\bm{N}=-1|\bm{uN}=-1)$ that a $1$ is flipped to a $-1$ and $\Pr(\bm{N}=-1|\bm{uN}=1)$ that a $-1$ is flipped to a $1$.
\end{remark}

Analogous to Theorem \ref{teo:additive equi}, the noise, $\bm{N}$, of the multiplicative wiretap channel, in Theorem \ref{teo:mult equivalence}, depends on the data $\bm{x}$ and the functions $f$ and $g$. In the next theorem, we show that if $\bm{x}$ is a well behaved random variable and $f$ and $g$ are low influence functions, i.e. their values do not depend too much on any coordinate, then $N(\bm{x})$ can be approximated by some noise, $N(\bm{g})$, which only depends on the functions $f$ and $g$ and is independent of $\bm{x}$.

\begin{restate} \label{teo: mult invariance}
Let $\bm{x} \sim \{-1 , 1\}^n$ satisfy Hypothesis \ref{hyp:var} and $\bm{g} = (\bm{g}_1, \ldots, \bm{g}_n)$ be such that each $\bm{g}_i$ is a standard Gaussian. Let $f,g:\{-1,1\}^n \rightarrow \{-1,1\}$ be of degree $k_1$ with $l_1$ terms and $k_2$ with $l_2$ terms with both $\Inf_t [f]$ and $\Inf_t [g]$ smaller than $\epsilon$, for every $t \in [n]$. Assume $\psi: \mathbb{R} \rightarrow \mathbb{R}$ is $\mathcal{C}^4$ with $||\psi''''||_\infty \leq C$. Then the noise $N=fg$ satisfies
\[ | E[\psi(N(\bm{x}))] - E[\psi(N(\bm{g}))] | \leq \frac{C}{3}  kl 9^k  \epsilon \]
where $k=k_1 k_2$ and $l = l_1 l_2$.
\end{restate}

\begin{proof}
Every Boolean function has $\Var[N] \leq 1$. Since $\Inf_t [f]$ and $\Inf_t [g]$ are smaller than $\epsilon$, Lemma \ref{lem: influence mult} implies in $\Inf_t [N]\leq 4\epsilon l$. Our result then follows from Corollary \ref{cor:bip}. 
\end{proof}

\begin{example} \label{ex:BAC}
Consider the computational wiretap in Figure \ref{fig:comp} where $f,g: \{1,-1 \}^3 \rightarrow \{1,-1 \}$ are such that ${f(x) = x_1 x_2 x_3}$ and
\begin{multline*}
  g(x) =\frac{1}{4} ( 1 - x_1 - x_2 - x_3 + x_1 x_2 + x_1 x_3  \\ +x_2 x_3 + 3 x_1 x_2 x_3 ).
\end{multline*}
Then, by Theorem \ref{teo:mult equivalence}, the computational wiretap is equivalent to the multiplicative wiretap channel in Figure \ref{fig:mult} with noise $N(\bm{x})=f(\bm{x})g(\bm{x})$. This is equivalent to the binary asymmetric channel, in this case a Z-channel, in Figure \ref{fig:BAC}.
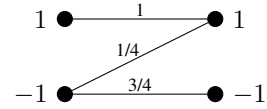
\begin{figure}[h]
\center
\tikzstyle{int}=[draw,shape=circle,fill=black,minimum size=0.2cm,inner sep=0pt]
\tikzstyle{init} = [pin edge={to-,thin,black}]

\begin{tikzpicture}[node distance=1cm,auto,>=latex']

	\node [int,label=left:$1$] (a)  {};
    \node [int,label=right:$1$] (b) [right of=a, node distance=2cm]  {};
    \node [int,label=left:$-1$] (c) [below of=a] {};
    \node [int,label=right:$-1$] (d) [below of=b]  {};
    
    \path[-] (a) edge node [inner sep=1pt] {\scriptsize{1}} (b);
    \path[-] (c) edge node [inner sep=0pt] {\scriptsize{1/4}} (b);
    \path[-] (c) edge node [inner sep=1pt] {\scriptsize{3/4}} (d);
    
\end{tikzpicture}
\caption{Binary asymmetric channel for Example \ref{ex:BAC}.} \label{fig:BAC}
\end{figure}

This noise depends not only on $f$ and $g$ but also on $\bm{x}$. Using Theorem $\ref{teo: mult invariance}$, we can approximate this noise by $N(\bm{g})$ which is independent of $\bm{x}$. 

One can check that both $\Inf_t [f]$ and $\Inf_t [g]$ are smaller than $1$, $k=9$ and $l=8$. Thus,
\[ | E[\psi(N(\bm{x}))] - E[\psi(N(\bm{g}))] | < 10^5 C .\]

\end{example}

\end{document}